\documentclass[11pt]{article}

\usepackage{fullpage}

\usepackage[utf8]{inputenc}

\usepackage{amsthm, amsmath, amssymb, amsfonts}

\usepackage[pdfencoding=auto]{hyperref}
\usepackage{algorithmic}
\usepackage[linesnumbered,ruled,lined,commentsnumbered]{algorithm2e}

\usepackage{thm-restate}

\usepackage{graphicx}
\graphicspath{{figs/}}
\usepackage{caption}
\usepackage{subcaption}
\usepackage[capitalise]{cleveref}

\usepackage{enumitem}   
\setlist{topsep=3pt, itemsep=3pt}

\usepackage{xspace}

\usepackage{float}
\usepackage[normalem]{ulem}

\usepackage{verbatim}

\newtheorem{theorem}{Theorem}[section]
\newtheorem{lemma}[theorem]{Lemma}
\newtheorem{observation}[theorem]{Observation}
\newtheorem{corollary}[theorem]{Corollary}

\usepackage{enumitem}
  
\newcommand{\dist}{\mathsf{dist}}

\newcommand{\Path}{\mathsf{Path}}

\newcommand{\NP}{\mathit{NP}}

\newcommand{\eT}{\mathsf{eT}}
\newcommand{\proc}{\mathsf{proc}}
\newcommand{\prev}{\mathsf{prev}}

\newcommand{\bbR}{{\mathbb{R}}}

\newcommand{\DD}{\mathsf{DD}} 
\newcommand{\DDT}{\mathsf{DDT}} 
\newcommand{\DDC}{\mathsf{DDC}} 
\newcommand{\N}{\mathsf{N}} 
\newcommand{\E}{\mathsf{E}}

\title{Package Delivery Using Drones with Restricted Movement Areas}
\author{
Thomas Erlebach\\ Durham University, UK \\
\texttt{ thomas.erlebach@durham.ac.uk}
        \and 
        Kelin Luo\\ University of Bonn, Germany\\
        \texttt{ kluo@uni-bonn.de } 
        \and Frits C.R. Spieksma \\  Eindhoven University of Technology, The Netherlands\\
\texttt{f.c.r.spieksma@tue.nl} 
}

\begin{document}
    \maketitle

\maketitle
\begin{abstract}
For the problem of delivering a package from a source node to a
destination node in a graph using a set of drones, we study the
setting where the movements of each drone are restricted to a certain
subgraph of the given graph. We consider the objectives of
minimizing the delivery time (problem $\DDT$) and of minimizing the total energy
consumption (problem $\DDC$). For general graphs, we show a strong
inapproximability result and a matching approximation algorithm
for $\DDT$ as well as $\NP$-hardness and a $2$-approximation algorithm
for $\DDC$. For the special case of a path, we show that $\DDT$ is
$\NP$-hard if the drones have different speeds.
For trees, we give optimal algorithms under the assumption
that all drones have the same speed or the same energy consumption rate.
The results for trees extend to arbitrary graphs if the subgraph of
each drone is isometric.
\end{abstract}

    \thispagestyle{empty}

 
\section{Introduction}
\label{sec:intro}%
Problem settings where multiple drones collaborate to
deliver a package from a source location to a target
location have received significant attention recently.
One motivation for the study of such problems comes
from companies considering the possibility
of delivering parcels to consumers via drones, e.g., Amazon Prime Air~\cite{amazon22}.
In previous work in this area~\cite{Chalopin2013,Chalopin2014,DBLP:conf/stacs/BartschiC0D0HP17,DBLP:conf/mfcs/Bartschi0M18,DBLP:journals/tcs/BartschiCDDGGLM20,DBLP:journals/jcss/CarvalhoEP21,DBLP:journals/tcs/ChalopinDDLM21},
the drones are typically
modeled as agents that move along the edges of a graph,
and the package has to be transported from a source node
to a target node in that graph. Optimization objectives
that have been considered include minimizing the delivery
time, minimizing the energy consumption by the agents,
or a combination of the two. A common assumption has been
that every agent can travel freely throughout the whole
graph~\cite{DBLP:conf/stacs/BartschiC0D0HP17,DBLP:conf/mfcs/Bartschi0M18,DBLP:journals/jcss/CarvalhoEP21}, possibly with a restriction
of each agent's travel distance due to energy
constraints~\cite{Chalopin2013,Chalopin2014,DBLP:journals/tcs/BartschiCDDGGLM20,DBLP:journals/tcs/ChalopinDDLM21}.
In this paper, we study for the first time a variation
of the problem in which each agent is only allowed to travel
in a certain subgraph of the given graph.

We remark
that the previously considered problem in which each agent has an energy
budget that constrains its total distance
traveled~\cite{Chalopin2013,Chalopin2014,DBLP:journals/tcs/BartschiCDDGGLM20,DBLP:journals/tcs/ChalopinDDLM21} is fundamentally
different from the problem considered here in which each agent can only travel
in a certain subgraph: In our problem, an agent can still
travel an arbitrary distance by moving back and forth many times within
its subgraph. Furthermore, the subgraph in which an agent can travel
cannot necessarily be defined via a budget constraint. This means
that neither hardness results nor algorithmic results translate
directly between the two problems.

As motivation for considering agents with movement
restrictions, we note that in a realistic setting, the
usage of drones may be regulated by licenses that
forbid some drones from flying in certain areas.
The license of a drone operator may only allow that
operator to cover a certain area.
Furthermore, densely populated areas
may have restrictions on which drones are allowed
to operate there.
Package delivery with multiple
collaborating agents might also involve different types
of agents (boats, cars, flying drones, etc.), where it is natural
to consider the case that each agent can traverse only a
certain part of the graph: For example, a boat can only traverse
edges that represent waterways.

In our setting, we are given an undirected graph $G=(V,E)$
with a source node~$s$ and a target node~$y$ of the package
as well as a set of $k$ agents. The subgraph in which an
agent $a$ is allowed to operate is denoted by $G_a=(V_a,E_a)$.
Each agent can pick up the package
at its source location or from another drone, and it can deliver
the package to the target location or hand it to another drone.
We consider the objective of minimizing the delivery time,
i.e., the time when the package reaches~$y$, as well as the
objective of minimizing the total energy consumption of the drones.
We only consider the problem for a single package.

\subparagraph*{Related work.}
Collaborative delivery of a package from a source node $s$
to a target node $y$ with the goal of minimizing the delivery
time was considered by B\"artschi et al.~\cite{DBLP:conf/mfcs/Bartschi0M18}. For $k$ agents
in a graph with $n$ nodes and $m$ edges, they showed that
the problem can be solved in $O(k^2 m+kn^2+\mathrm{APSP})$
time for a single package, where $\mathrm{APSP}$ is the
time for computing all-pairs shortest paths in a graph
with $n$ nodes and $m$ edges. Carvalho et
al.~\cite{DBLP:journals/jcss/CarvalhoEP21} improved the
time complexity for the problem to $O(kn \log n + km)$
and showed that the problem is $\NP$-hard if two packages
must be delivered.

For minimizing the energy consumption for the delivery
of a package, B\"artschi et al.~\cite{DBLP:conf/stacs/BartschiC0D0HP17}
gave a polynomial algorithm for one package and showed that
the problem is $\NP$-hard for several packages.

The combined optimization of delivery time $\mathcal{T}$
and energy consumption $\mathcal{E}$ for the collaborative
delivery of a single package has also been considered.
Lexicographically minimizing $(\mathcal{E},\mathcal{T})$ 
can be done in polynomial time~\cite{DBLP:conf/fct/BartschiT17}, but lexicographically minimizing
$(\mathcal{T},\mathcal{E})$ or minimizing
any linear combination of $\mathcal{T}$ and $\mathcal{E}$
is $\NP$-hard~\cite{DBLP:conf/mfcs/Bartschi0M18}.

Delivering a package using $k$ energy-constrained agents
has been shown to be strongly $\NP$-hard in general graphs~\cite{Chalopin2013}
and weakly $\NP$-hard on a path~\cite{Chalopin2014}.
B\"artschi et al.~\cite{DBLP:journals/tcs/BartschiCDDGGLM20} showed
that the variant where each agent must return to its
initial location can be solved in polynomial
time for tree networks. The variant in which the package
must travel via a fixed path in a general graph
has been studied by Chalopin
et al.~\cite{DBLP:journals/tcs/ChalopinDDLM21}.

\subparagraph*{Our results.}
In Section~\ref{sec:pre}, we introduce definitions and give some auxiliary results.
In Section~\ref{sec:hardness}, we show that movement restrictions make the drone delivery
problem harder for both objectives: For minimizing the
delivery time, we show that the problem is $\NP$-hard to
approximate within ratio $O(n^{1-\epsilon})$ or $O(k^{1-\epsilon})$.
For minimizing the energy consumption, we show $\NP$-hardness.
These results hold even if all agents have the same speed and
the same energy consumption rate.

In Section~\ref{sec:alg_DDT}, we propose an
$O(\min\{n,k\})$-approximation algorithm for the problem
of minimizing the delivery time. The algorithm first
computes a schedule with minimum delivery time for the
problem variant where an arbitrary number of copies
of each agent is available. Then it transforms the
schedule into one that is feasible for the
problem with a single copy of each agent. The algorithm
can also handle handovers on edges.
In Section~\ref{sec:alg_DDC}, we give a
$2$-approximation algorithm for the problem of
minimizing the total energy consumption.
We again first compute an optimal schedule that may use
several copies of each agent and then transform the
schedule into one with a single copy of each agent.

In Section~\ref{sec:line},
we first consider the special case where the graph is a path (line) 
and the subgraph of each agent is a subpath. For this case,
we show that the problem of minimizing the delivery time
is $\NP$-hard if the agents have different speeds. If the agents
have the same speed or the same energy consumption rate,
we show that the problem of minimizing the delivery time or the problem of minimizing
the total energy consumption are both polynomial-time solvable even
for the more general case when the graph is a tree, or when the
graph is arbitrary but the subgraph of every agent is isometric
(defined in Section~\ref{sec:line-alg}).
Conclusions are presented in Section~\ref{sec:conclusion}. 

\section{Preliminaries}
\label{sec:pre}
We now define the \emph{drone delivery} ($\DD$) problem formally. We are given a connected graph $G=(V, E)$ with edge lengths $\ell : E \to \bbR_{\ge 0}$, where $\ell(u, v)$ represents the distance between $u$ and $v$ along edge $\{u,v\}$.
(We sometimes write $G=(V,E,\ell)$ to include an explicit reference to~$\ell$.)
Let $n=|V|$ and $m=|E|$. We are also given a set $A$ containing $k\ge 1$ mobile agents (representing drones).
Each agent $a\in A$ is specified by $a=(p_a, v_a, w_a, V_a, E_a)$, where $p_a\in V $ is the agent's \emph{initial position},  and $v_a> 0$ and $w_a\ge 0$ are the agent's \emph{velocity} (or speed) and \emph{energy consumption rate}, respectively. 
To traverse a path of length~$x$, agent $a$ takes time  $x/ v_a$ and consumes $x\cdot w_a$ units of energy. Furthermore, $V_a\subseteq V $ and $ E_a\subseteq E$ are the \emph{node-range} and \emph{edge-range} of
agent~$a$, respectively. Agent~$a$ can only travel to/via nodes in $V_a$ and edges in $E_a$. We require $p_a\in V_a$. 
To ensure meaningful solutions, we make the following two natural assumptions:  
\begin{itemize} 
\item  For each agent $a$, the graph $G_a=(V_a, E_a)$ is a connected subgraph of $G$.
\item The union of the subgraphs  $G_a$ over all agents is the graph $G$, i.e., $\bigcup_{a\in A} V_a = V$ and $\bigcup_{a\in A} E_a = E$.  This implies that there is a feasible schedule for any package $(s, y)$.
 
\end{itemize}  
The package is specified by $(s, y)$, where $s, y\in V$ are the \emph{start node} (start location) and \emph{destination node} (target location), respectively.  
The task is to find a schedule for delivering the package from the start node to the destination node while achieving a specific objective.
The problem of minimizing the delivery time is denoted by~$\DDT$, and
the problem of minimizing the consumption by~$\DDC$.  
 

To describe solutions of the problems, we need to define how
a \emph{schedule} can be represented. A schedule is given as a list of \emph{trips} $\mathcal{S}=\{S_1, S_2, \ldots, S_h\}$: 
$$\{(a_1, t_1, \langle o_1, \ldots, u_0  \rangle,  \langle u_0, \ldots, u_1 \rangle),  \ldots,  (a_h, t_h, \langle o_h, \ldots, u_{h-1} \rangle,  \langle u_{h-1}, \ldots, u_h \rangle)\}$$ 
The $i$-th trip $S_i=(a_i, t_i, \langle o_i, \ldots, u_{i-1} \rangle,  \langle u_{i-1}, \ldots, u_i \rangle)$ represents two consecutive trips taken by agent $a_i$ starting at time $t_i\ge 0$: an \emph{empty movement} trip
traversing nodes $O_i=\langle o_i, \ldots, u_{i-1} \rangle$, and a \emph{delivery trip} (during which $a_i$ carries the package) of traversing nodes $U_i= \langle u_{i-1}, \ldots, u_{i} \rangle$. The agent $a_i$ picks up the package at node $u_{i-1}$ and drops it off at node $u_i$. With slight abuse of notation, we also use $O_i$ and $U_i$ to denote the set of edges in each of these two trips. If $S_i$ is the first trip of agent $a_i$, then
$o_i$ is the agent's initial location. Otherwise, $o_i$ is the
location where the agent dropped off the package at the end of
its previous trip.
Initially, each agent $a\in A$ is at node $p_a$ at time~$0$.
In the definition of schedules, when we allow two agents to
meet on an edge to hand over the package, we allow the nodes also
to be points on edges: For example, a node $u_i$ could be the point on an edge
$\{u,v\}$ with length $5$ that is at distance $2$ from $u$ (and hence
at distance $3$ from $v$).

Let $T(u_i)$ (resp.\ $C(u_i)$) denote the time passed
(resp.\ the energy consumed) until the package is delivered to node $u_i$ in schedule $\mathcal{S}$, i.e.,
$$T(u_i)= \max\bigm\{T(u_{i-1}), t_i+\sum_{e\in O_i} \frac{\ell(e)}{v_{a_i}}  \bigm\} +\sum_{e \in U_{i}} \frac{\ell(e)}{v_{a_i}}, $$
$$C(u_i)= C(u_{i-1}) + w_{a_i} \cdot  \sum_{e\in O_i} \ell(e) + w_{a_i}\cdot \sum_{e \in U_{i}} \ell(e). $$ 
In the formula for $T(u_i)$, the first term represents the earliest time
when agent $a_i$ can pick up the package at $u_{i-1}$: The package reaches
$u_{i-1}$ at time $T(u_{i-1})$, and the agent $a_i$ reaches $u_{i-1}$
at time $t_i+\sum_{e\in O_i}(\ell(e)/v_{a_i})$. Hence, the maximum of these
two times is the time when both the agent and the package have reached
$u_{i-1}$. The second term in the formula represents the time that
agent~$a_i$ takes, once it has picked up the package at $u_{i-1}$,
to carry it along $U_i$ to $u_i$.
The three terms in the formula for $C(u_i)$ represent the energy consumed
until the package reaches $u_{i-1}$, the energy consumed by agent~$a_i$
to travel from $o_i$ to $u_{i-1}$ along $O_i$, and the energy consumed
by agent $a_i$ to carry the package from $u_{i-1}$ to $u_i$ along $U_i$,
respectively.

The pick-up location of the first agent must be the start node, i.e., $u_0=s$, and the drop-off location of the last agent must be the destination node, $u_{h}=y$. 
We let $T(u_0)= 0 $ and
$C(u_0)=0$. 
The goal of the $\DDT$ problem is to find a feasible schedule  $\mathcal{S}$ that
minimizes the delivery time $T(y)$, and the goal of the $\DDC$ problem is to find a feasible schedule $\mathcal{S}$ that
minimizes the energy consumption $C(y)$.

So far, we have defined the $\DDT$ and $\DDC$ problem. As in previous studies~\cite{DBLP:journals/jcss/CarvalhoEP21}, we further distinguish variants of these problems based on the \emph{handover manner}. 

\textbf{Handover manner.} 
The handover of the package between two agents may occur at a node or at some interior point of an edge. When the handovers are restricted to be on nodes, we say the drone delivery problem is handled with \emph{node handovers}, and when the handover can be done on a node or at an interior point of an edge, we say that the drone delivery problem is handled with \emph{edge handovers}. We use the subscripts $\N$ and $\E$ to represent node handovers and edge handovers, respectively. Thus, we get four variants of the drone delivery problem: $\DDT_\N$, $\DDT_\E$, $\DDC_\N$ and $\DDC_\E$.   

\textbf{With or without initial positions.} We additionally consider problem variants in which the initial positions of the agents are not fixed (given), which means that the initial positions $p_a$ for ${a\in A}$ can be chosen by the algorithm. When the initial positions are fixed, we say that the problem is \emph{with initial positions}, and when the initial positions are not fixed, we say that the problem is \emph{without initial positions}. When we do not specify that a problem is without initial positions, we always refer to the problem with initial positions by default.

In the rest of the paper, in order to simplify notation, the $i$-th trip in a schedule $\mathcal{S}$ is usually written in simplified form as $(a_i, u_{i-1}, u_{i})$, where $a_i$ is the agent, $u_{i-1}$ is the pick-up location of the package, and $u_{i}$ is the drop-off location of the package. We can omit the agent's empty movement trip (including its previous position) and its start time because the agent $a_i$ always takes the path in $G_{a_i}=(V_{a_i},E_{a_i})$ with minimum cost (travel time or energy consumption) from its previous position to the current pick-up location $u_{i-1}$ and then from $u_{i-1}$ to $u_i$.

For each node $v\in V$, we use $A(v)$ to denote the set of agents that can visit the node~$v$; and for every edge $\{u,v\}$ in $E$, we use $A(u, v)$ to denote the set of agents that can traverse the edge $\{u,v\}$.
For any pair of nodes $u, v \in V_a$ for some $a\in A$, we denote by $\dist_a(u,v)$ the length of the shortest path between node $u$ and $v$ in the graph $G_a=(V_a, E_a)$.

\subparagraph*{Useful properties.}
We show some basic properties for the drone delivery problem.
The following lemma can be shown by an exchange argument: If an agent was involved
in the package delivery at least twice, letting that agent keep the package from its first involvement to its last involvement does not increase the delivery time nor the energy consumption.

\begin{restatable}{lemma}{lemminConce}
\label{lem:minT-once}\label{lem:minC-once}
For every instance of $\DDT_\N$, $\DDT_\E$, $\DDC_\N$ and $\DDC_\E$ (with or without initial positions), there is an optimal solution in which each of the involved agents picks up and drops off the package exactly once. 
\end{restatable}

\begin{proof}
For an optimal solution, we label the agents $1,2,\ldots$ in the order
in which they carry the package in the schedule $\mathcal{S}$.
Agent $i$ and agent $i+1$ are different for each $i$. We show that there is an optimal solution in which each of the involved agents picks up and drops off the package exactly once.

Assume that there is an agent involved in the optimal schedule $\mathcal{S}$ more than once.  
%
Let $i$ be the first agent in $\mathcal{S}$ such that agent $j>i+1$ and agent $i$ are the same agent. Then we can simply remove the trips of agents $i+1,i+2,..., j-1$ and let agent $i$ keep the package and
travel the edges that agent $i$ originally traversed after dropping off the package at $u_i$ to reach $u_{j-1}$. This is because agent $i$ can reach the drop-off location $u_{j-1}$ of agent $j-1$ before agent $j$ picks up the package at $u_{j-1}$. By doing so, we do not increase the time between the drop-off time of agent $i$ and the pick-up time of agent $j$. Agent~$i$ (equal to agent $j$) can then carry the package further to $u_j$. Furthermore, the later agents $j+1, j+2, \ldots$ can deliver the package following the original schedule on time. By repeating this modification, we can produce an optimal schedule for both $\DDT_\N$ and $\DDT_\E$ such that each agent is involved in the delivery of the package exactly once.

Similarly, it is easy to show the energy consumption does not increase by the above modification. This gives the proof for $\DDC_\N$ and $\DDC_\E$. 
\end{proof}

For $\DDC$, we can show that handovers at interior points of edges cannot
reduce the energy consumption. If two agents carry the package consecutively
over parts of the same edge, letting one of the two agents carry the package
over those parts can be shown not to increase the energy consumption.

\begin{restatable}{lemma}{lemeqMinCSESN}
\label{lem:eqMinC-SE-SN}
For any instance of $\DDC_\N$ and $\DDC_\E$ (with or without initial positions), there is a solution
that is simultaneously optimal for both problems. In other words, there is an optimal solution for $\DDC_\E$ in which all handovers of the package take place at nodes of the graph.
\end{restatable}

\begin{proof}
Consider an optimal solution of $\DDC_\E$ that minimizes the number of handovers that take place on edges of the graph.
Let~$e=\{u,v\}$ be an edge on which at least one handover takes place in that solution, and let the package travel from $u$ to~$v$.
Let $a_1$ and $a_2$ be the first two agents carrying the package on~$e$, with $a_i$ carrying the package over distance $\ell_i>0$ for $i=1,2$.
Let $v_1$ denote the point on $e$ at distance $\ell_1$ from $u$ and let $v_2$ denote the point on $e$ at distance $\ell_1+\ell_2$ from $u$.
Note that $v_2=v$ is possible.
Agent $a_1$ carries the package from $u$ to $v_1$, and agent $a_2$ carries the package from $v_1$ to $v_2$.

First, consider the problem variants with initial positions.
Agent~$a_2$
must reach $v_1$ coming from~$v$, because otherwise (if agent~$a_2$
reaches $v_1$ coming from $u$) we could let $a_2$ carry the package
from $u$ to $v_1$ and save $a_1$'s trip from $u$ to $v_1$, reducing
the energy consumption.
Therefore, $a_2$~travels distance $\ell_2$ twice: First to go from $v_2$ to $v_1$ without the package, then to carry the package from $v_1$ to $v_2$.
Hence, the two agents consume energy $w_1 \ell_1 + 2 w_2 \ell_2$ while traveling between $u$ and $v_2$.

Consider the following two alternatives where only one of the two agents is used to carry the package from $u$ to $v_2$:
Agent $a_1$ could bring the package from $u$ to $v_2$, consuming energy $w_1(\ell_1+\ell_2)$, or agent $a_2$ could travel to $u$ to pick up the package and then bring it to $v_2$, consuming energy $2 w_2(\ell_1+\ell_2)$.
We claim that at least one of these alternatives does not consume more energy than the original solution that uses both agents.
Assume otherwise. Then we have:
\begin{align*}
    w_1 \ell_1 + 2 w_2 \ell_2 &<  w_1(\ell_1+\ell_2) \\
    w_1 \ell_1 + 2 w_2 \ell_2 &<  2 w_2(\ell_1+\ell_2)
\end{align*}
These inequalities can be simplified to:
\begin{align*}
    2 w_2 \ell_2 &<  w_1\ell_2 \\
    w_1 \ell_1 &<  2 w_2\ell_1
\end{align*}
As $\ell_1,\ell_2>0$, we get that $2w_2<w_1$ and $w_1<2 w_2$, which is a contradiction.
Thus, we can remove the handover between $a_1$ and $a_2$ at $v_1$ and use only one of these two agents to transport the package from $u$ to $v_2$, without increasing the energy consumption.
This is a contradiction to our assumption that we started with an optimal solution that minimizes the number of handovers on edges.
Hence, there is an optimal solution to $\DDC_\E$ in which no handovers take place on edges, and that solution is also an optimal solution to $\DDC_\N$.

Now, consider the problem variants without initial positions.
The energy consumed by agents $a_1$ and $a_2$ to bring the package from $u$ to $v_2$ is $w_1\ell_1+w_2\ell_2$ (the algorithm can choose $v_1$ as the initial position of agent $a_2$). If $w_1\le w_2$, we can modify the solution and let $a_1$ carry the package from $u$ to $v_2$. If $w_1>w_2$, we can choose $u$ as initial position for $a_2$ and let $a_2$ carry the package from $u$ to $v_2$. Hence, we can remove the handover at $v_1$ without increasing the energy consumption, and the result follows as above.
\end{proof}

For the case that $G_a=G$ for all $a\in A$, it has
been observed in previous work \cite{DBLP:conf/mfcs/Bartschi0M18, DBLP:journals/jcss/CarvalhoEP21} that there is an optimal schedule
for $\DDT_\N$ and $\DDT_\E$ in which the velocities of the
involved agents are strictly increasing, and an optimal schedule
for $\DDC_\N$ and $\DDC_\E$ in which the consumption rates of the
involved agents are strictly decreasing. We remark that this
very useful property does not necessarily hold in our setting
with agent movement restrictions. This is the main reason
why the problem becomes harder, as shown by our hardness results
in Section~\ref{sec:hardness} and, even for path networks,
in Section~\ref{sec:line-hardness}.

\section{Hardness results}
\label{sec:hardness}%
In this section, we prove several hardness results for $\DDT$ and $\DDC$
via reductions from the $\NP$-complete \emph{$3$-dimensional matching problem} (3DM)~\cite{DBLP:books/fm/GareyJ79}.
All of them apply to the case where all agents have the same speed and the same energy consumption rate. We first present the construction of a \emph{base instance} of $\DDT$, which we then adapt to obtain the hardness results for $\DDT$ and $\DDC$.
The instances that we create have the property that
every edge $e$ of the graph is only in the set $E_a$ of a single agent~$a$.
Therefore, handovers on edges are not possible for these instances, and so
the variants of the problems with node handovers and with edge handovers are equivalent for these
instances.
Hence, we only need to consider the problem variant with node handovers in the proofs.



The problem 3DM is defined as follows:
Given are three sets $X,Y,Z$, each of size $n$,  and a set $\mathcal{F}\subseteq X\times Y \times Z$ consisting of $m$ triples. Each triple in $\mathcal{F}$
is of the form $(x, y, z)$ with $x \in X$, $y \in Y$, and $z \in Z$. The question is: Is there a set of $n$ triples in $\mathcal{F}$ such that each element of $X\cup Y\cup Z$ is contained in exactly one of these $n$ triples?
 
 Let an instance of 3DM be given by $X,Y,Z$ and $\mathcal{F}$.
 Assume $X=\{x_1,x_2,\ldots,x_n\}$,
 $Y=\{y_1,y_2,\ldots,y_n\}$ and $Z=\{z_1,z_2,\ldots,z_n\}$.
 Number the triples in $\mathcal{F}$ from $1$ to $m=|\mathcal{F}|$ arbitrarily,
 and let the $j$-th triple be $t_j=(x_{f(j)},y_{g(j)},z_{h(j)})$
 for suitable functions $f,g,h:[m]\rightarrow [n]$.
 We create from this a base instance $I$
 of $\DDT_\N$ that consists of $n$ selection gadgets
 and $3n$ agent gadgets. Carrying the package through a selection gadget
 corresponds to selecting one of the $m$ triples of $\mathcal{F}$. The $n$ selection
 gadgets are placed consecutively, so that the package travels through
 all of them (unless it makes a detour that increases the delivery time).
 For each element $q$ of
 $X\cup Y\cup Z$, there is an agent gadget containing the start position
 of a unique agent. The agent gadget for element $q$ ensures that the agent
 can carry the package on the edge corresponding to element $q$ on a path through a selection gadget if and only if that path corresponds to a triple that contains~$q$.
 If the instance of 3DM is a yes instance, then each of the agents from the agent gadgets only needs to transport the package on one edge.
 Otherwise, an agent must travel (with or without the package) from one selection gadget to another one, and the extra time consumed by this movement increases the delivery time.
 
 Formally, the instance of $\DDT_\N$ with graph $\bar{G}=(\bar{V}, \bar{E})$ is created as follows. See Figure~\ref{fig:grid} for an illustration.
 The vertex set $\bar{V}$ contains $3n$ \emph{agent nodes} and $n(3m+1)+1$ nodes in selection gadgets, a total of $4n+3mn+1$ nodes.
 The agent nodes comprise
 a node $x_i$ for each $x_i\in X$, a node $y_i$ for each $y_i\in Y$, and
 a node $z_i$ for each $z_i\in Z$.
 The node set of selection gadget~$i$, for $1\le i\le n$,
 is $\{s_i,s_{i+1}\}\cup\{v^{i,x}_j,v^{i,y}_j,v^{i,z}_j\mid 1\le j\le m\}$,
 a set containing $3m+2$ nodes. For $1\le i<n$, the node $s_{i+1}$ is contained
 in selection gadget $i$ and in selection gadget $i+1$.
 The path through the three nodes $v^{i,x}_j,v^{i,y}_j,v^{i,z}_j$ (in selection gadget $i$ for any $i$) represents the triple $t_j$ in $\mathcal{F}$.
 The package is initially located at node $s_1$ and must be delivered to node $s_{n+1}$.
 

The edge set $\bar{E}$ contains two types of edges: \emph{inner} edges, each of which connects two nodes in a selection gadget, and \emph{outer} edges, each of which connects an agent node and a node in a selection gadget.
For every $i\in[n]$ and $j\in[m]$, there are four inner edges $ \{s_i, v^{i,x}_j\}$, $\{v^{i,x}_j, v^{i,y}_j\}$, $\{v^{i,y}_j, v^{i,z}_j\}$ and $\{v^{i,z}_j, s_{i+1}\}$ of length~$0$.\footnote{We could set the lengths of these edges to a small $\epsilon>0$ if we wanted to avoid edges of length~$0$.}
So there are $4mn$ inner edges. The outer edges are as follows: For every $i \in [n]$ and $j \in [m]$, there are the three outer edges
$\{x_{f(j)},v^{i,x}_j\}$,
$\{y_{g(j)},v^{i,y}_j\}$,
$\{z_{h(j)},v^{i,z}_j\}$ of length $M$ for some $M>0$. (We can set $M=1$.)


 This completes the description of the graph ${\bar G} = ({\bar V}, {\bar E})$ with edge lengths.

\begin{figure}[t]
\centering
\includegraphics[width=\textwidth]{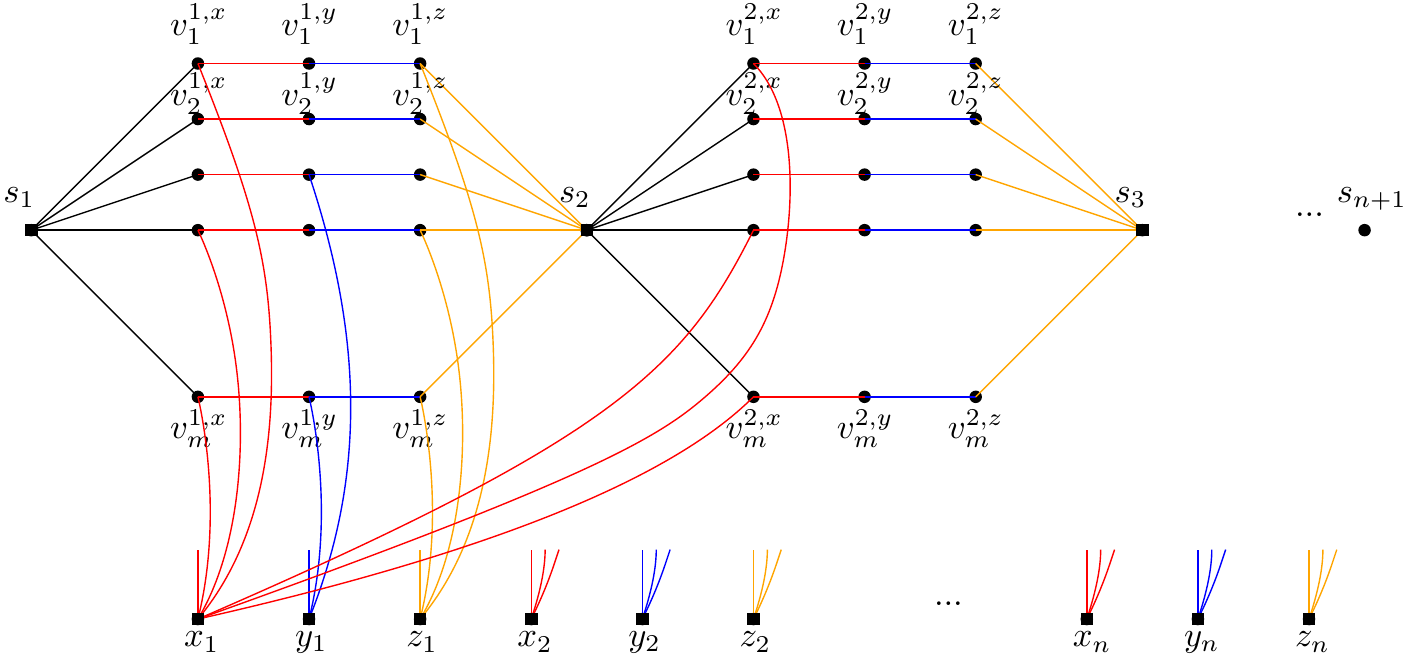}
\caption{The graph ${\bar G} = ({\bar V}, {\bar E})$ used for the inapproximability and hardness proofs for $\DDT$ and $\DDC$. The inner edges have length $0$, and the outer edges have length $M$. Each square node represents a drone's initial location.  
Each colored (red, blue, and orange) path represents a set of edges that can be visited by the same agent which is initially located at a location in $\{x_i,y_i,z_i\mid i\in [n]\}$. For each $i\in [n]$, there is an additional agent that is initially located at $s_i$ and which can visit edges $\{\{s_i, v^{i, x}_x\} \mid j\in [m]\}$.} 
\label{fig:grid}
\end{figure}

Now we define the set of agents. We have $4n$ agents with unit speed and unit energy consumption rate. Their initial locations are
$\{x_i, y_i, z_i, s_i \mid i\in [n]\}$.
The agents with initial location in $\{x_i,y_i,z_i\mid i\in [n]\}$ are
called \emph{element agents}.
For each $i\in [n]$, the node ranges and edge ranges of the agents with initial location with subscript $i$ are defined as follows:

\begin{itemize}
	\item The agent located at $x_i$ has node range $\{x_i, v^{i',x}_{j}, v^{i', y}_{j} \mid i'\in [n], j\in [m], f(j)=i \}$
	and edge range $\{\{x_i, v^{i',x}_{j}\}, \{v^{i',x}_{j}, v^{i', y}_{j}\}\mid i'\in [n], j\in [m], f(j)=i \}$.
	\item  The agent located at $y_i$ has node range $\{y_i, v^{i',y}_{j}, v^{i', z}_{j} \mid i'\in [n], j\in [m], g(j)=i \}$
	and edge range $\{\{y_i, v^{i',y}_{j}\}, \{v^{i',y}_{j}, v^{i', z}_{j}\}\mid i'\in [n], j\in [m], g(j)=i \}$.
	\item The agent located at $z_i$ has node range $\{z_i, v^{i',z}_{j}, s_{i'+1} \mid i'\in [n], j\in [m], h(j)=i \}$
	and edge range $\{\{z_i, v^{i',z}_{j}\}, \{v^{i',z}_{j}, s_{i+1}\}\mid i'\in [n], j\in [m], h(j)=i \}$.
	\item The agent located at $s_i$ has node range $\{ s_i,  v^{i,x}_{j}\mid j\in [m] \}$ and edge range $\{\{s_i, v^{i,x}_{j}\} \mid j\in [m] \}$.
	%
	%
\end{itemize}\medskip

We now show that the given instance of 3DM is a yes-instance if and only if
the constructed base instance of $\DDT_\N$ has an optimal schedule with delivery time~$M$.
%
%
%
%
First, assume that the given instance of 3DM is a yes-instance,
and let $\{t_{k_1},t_{k_2},\ldots,t_{k_n}\}\subseteq \mathcal{F}$ be a perfect matching. We let the package travel from $s_1$ to $s_{n+1}$ via the $n$ selection gadgets, using the path via nodes $v^{i,x}_{k_i}, v^{i,y}_{k_i}, v^{i,z}_{k_i}$ in selection gadget~$i$, for $i\in [n]$. This path consists of $4n$ edges. Each of the $4n$ agents needs to carry the package on exactly one edge of this path. All the element agents reach the node where they pick up the package at time~$M$. As all edges in the selection gadgets have length~$0$,
this shows that the package reaches $s_{n+1}$ at time~$M$. It is clear that this
solution is optimal because at least one element agent must take part in the delivery and cannot pick up the package before time $M$.

Now, assume that the base instance of $\DDT_\N$ admits a schedule with delivery time $M$. It is not hard to see that
the schedule must be of the above format, using each agent on exactly one edge in one selection gadget. This is because it takes time $2M$ for an element agent to move from one selection gadget to another one via its initial location.

This shows that there is a schedule with delivery time $M$ if and only if
the given instance of 3DM is a yes-instance. Furthermore, as any element
agent needs time $M$ to reach a pick-up point in a selection gadget and time $2M$ to move to
a different selection gadget (with or without the package), it is clear that
the optimal schedule will have length at least $3M$ if the given instance
of 3DM is a no-instance. This already shows that $\DDT_\N$ is $\NP$-hard and does
not admit a polynomial-time approximation algorithm with approximation
ratio smaller than~$3$ unless $P=\NP$, but we can strengthen the inapproximability
result by concatenating multiple copies of the base instance, as we show in the following theorem. 


\begin{restatable}{theorem}{LBminT} 
\label{the:lowerboundminT}
 For any constant $\epsilon> 0$, it is $\NP$-hard to approximate $\DDT_\N$ (or $\DDT_\E$) within a factor of $O(\min\{n^{1-\epsilon}, k^{1-\epsilon}\})$
 even if all agents have the same speed.
\end{restatable}

\begin{proof}
Consider the base instance $\bar{G}$ constructed from an instance $I$ of 3DM with
$|X|=|Y|=|Z|=n$ and $|\mathcal{F}|=m$.
To obtain a stronger inapproximability result for $\DDT_\N$,
we concatenate $q$ copies
of $\bar{G}$, for a suitable choice of $q$ (polynomial in $n$ and $m$).
The node $s_{n+1}$ of each copy of $\bar{G}$ (except the last) is identified with
the node $s_{1}$ of the next copy of $\bar{G}$. The package is initially located at
the node $s_1$ of the first copy and needs to be delivered to node $s_{n+1}$
of the last copy. We call the initial location of the package~$s$ (the node $s_1$ of the first copy) and
the target location~$y$ (the node $s_{n+1}$ of the last copy). The resulting instance of $\DDT_\N$ has
$N=(4n + 3mn)q + 1$ nodes and $k=4nq$ agents.
Let $G=(V,E)$ denote the graph of the resulting instance.

If the instance $I$ of 3DM admits a perfect matching, then there is an optimal schedule for
the constructed instance of $\DDT_\N$ that has delivery time $M$: The element agents in all copies
of $\bar{G}$ reach the node at which they pick up the package at time~$M$.
Otherwise, the best schedule for the constructed instance of $\DDT_\N$
has delivery time at least $M+q \cdot 2M=(1+2q)M$: It takes time $M$ for the
first element agent to reach the location where it pick up the package; furthermore,
in each of the $q$ copies of $\bar{G}$, at least one agent will have to carry the package
over edges in two different selection gadgets, adding $2M$ to the delivery time.
This shows that it is $\NP$-hard to obtain an approximation ratio of at most~$2q$.

We choose $q=(4n+3mn)^c$ for an arbitrarily large constant~$c$.
Then the constructed instance of $\DDT_\N$ has
$N=(4n+3mn)q+1=(4n+3mn)^{c+1}+1$ nodes
and $k=4nq=4n(4n+3mn)^c$ agents.
As $q=(4n+3m)^c=(N-1)^{c/(c+1)}=\Omega(N^{1-1/(c+1)})$,
this shows that there is no $O(N^{1-\epsilon})$-approximation
algorithm unless $P=\NP$.
As $q=(4n+3m)^c\ge k^{c/(c+1)}=k^{1-1/(c+1)}$,
there is also no $O(k^{1-\epsilon})$-approximation algorithm unless $P=\NP$.
\end{proof}

To show $\NP$-hardness for $\DDC$, we adapt the base instance
by numbering the columns of the selection gadgets to which outer
edges are attached from $1$ to $3n$ (from left to right) and
letting the outer edges attached to column $i$ have length
$2^{3n+1-i}$. If the given instance of 3DM is a yes-instance,
exactly one outer edge attached to each column will be used,
for a total energy consumption of $2^{3n+1}-2$. Otherwise,
some column will be the first column in which an outer edge
is used twice, and the total energy consumption will be
at least $2^{3n+1}$.

\begin{restatable}{theorem}{LBminC} 
\label{the:lowerboundminC}
The problems $\DDC_\N$ (and $\DDC_\E$) are $\NP$-hard even if all
agents have the same energy consumption.
\end{restatable}

\begin{proof}
To prove the $\NP$-hardness of $\DDC_\N$, we construct the graph $G = (V, E)$ from the base instance $\bar{G}$ by changing the lengths of the outer edges:
For every $i \in [n]$ and $j \in [m]$, we let the outer edges $(x_{f(j)}, v^{i,x}_{j})$, $(y_{g(j)},v^{i,y}_j)$ and $(z_{h(j)},v^{i,z}_j)$ have lengths $2^{3n+1-3(i-1)-1}$,
$2^{3n+1-3(i-1)-2}$ and
$2^{3n+1-3(i-1)-3}$, respectively.
Intuitively, there are $3n$ columns of vertices in selection gadgets to which outer edges are attached (cf.~Figure~\ref{fig:grid}), and the lengths of the outer edges are powers of two that decrease column by column. The outer edges attached to the leftmost column (closest to $s_1$) have length $2^{3n}$, and those attached to the rightmost column (closest to $s_{n+1}$) have length~$2$. The outer edges attached to the $k$-th column, $1\le k\le 3n$, have length $2^{3n+1-k}$.
  
If the instance of 3DM admits a perfect matching, then there is a schedule for the constructed instance of $\DDC_\N$ in which every agent carries the package over exactly one edge of a path in exactly one selection gadget. That schedule has total energy consumption $2^{3n}+2^{3n-1}+\cdots +2=2^{3n+1}-2$.
Otherwise, at least one agent must carry the package over edges in at least two different selection gadgets. Assume that the first agent who does this picks up the package in column~$k$. Then that agent will have energy consumption at least $2\cdot 2^{3n+1-k}$ because it needs to go back to its initial location after picking up the package. The previously used agents (who brought the package from $s_1$ to a node in column~$k$) have energy consumption at least $2^{3n+1-1}+2^{3n+1-2}+\cdots +2^{3n+1-(k-1)}$.
The total energy consumption is at least
$\left(\sum_{t=1}^{k-1}2^{3n+1-t}\right)+2\cdot 3^{3n+1-k}
=2^{3n+1}$. We conclude that there is a schedule with total energy consumption
at most $2^{3n+1}-2$ if and only if the instance of 3DM is a yes-instance.
\end{proof}

Finally, for the problem variants without initial positions,
we observe that the base instance admits a delivery schedule
with delivery time~$0$ and energy consumption~$0$ if and only
if the given instance of 3DM is a yes-instance. This shows
that there cannot be any polynomial-time approximation algorithm
for any of the $\DDT$ and $\DDC$ problem variants without initial
positions. 


\begin{restatable}{theorem}{LBminTni} 
\label{the:lowerboundminTni}
 For $\DDT_\N$ (or $\DDT_\E$) without initial positions, it is
 $\NP$-hard to distinguish between instances with delivery time~$0$
 and instances with delivery time greater than~$0$, even if all
 agents have the same speed. Therefore,
 there is no polynomial-time approximation algorithm for these
 problems (with any finite ratio), unless $P=\NP$.
 \end{restatable}

\begin{proof}
For $\DDT_\N$ without initial positions, consider the base
instance. If the given instance of 3DM is a yes-instance,
there is a schedule in which the package travels along
a path of $4n$ edges, these edges have length~$0$, and each agent
needs to carry the package over exactly one of these $4n$ edges.
Hence, we can place each agent at the start of the edge over
which it must carry the package, and the package
can be delivered in time~$0$. If the instance of 3DM is
a no-instance, at least one element agent will have to
carry the package on edges in two different selection gadgets,
and that agent will require time $2M$ to travel to the second
selection gadget. Thus, the delivery time will be at least $2M$.
\end{proof}

\begin{restatable}{theorem}{LBminCni} 
\label{the:lowerboundminCni}
For $\DDC_\N$ (or $\DDC_\E$) without initial positions, it is
$\NP$-hard to distinguish between instances with energy consumption~$0$
and instances with energy consumption greater than~$0$, even if all
agents have the same energy consumption rate. Therefore,
there is no polynomial-time approximation algorithm for these
problems (with any finite ratio), unless $P=\NP$.
\end{restatable}

\begin{proof}
We again use the
base instance, and the argument is analogous to the proof of
Theorem~\ref{the:lowerboundminTni}: If the
instance of 3DM is a yes-instance, we can place each
agent at the start of the edge over which it needs to
carry the package, and the agent can carry the package
over that edge with energy consumption~$0$. Otherwise,
at least one agent will have to travel between two different
selection gadgets, consuming energy $2M$.
\end{proof}

\paragraph{Remark.}  
The results of Theorems~\ref{the:lowerboundminTni} and \ref{the:lowerboundminCni}
arise because we allow zero-length edges and it is $\NP$-hard to
decide whether an optimal solution with objective value~$0$ exists.
If we were to require strictly positive edge lengths,
it would be possible to obtain approximation algorithms with
ratios that depend on the ratio of maximum to minimum edge length.

\section{Approximation algorithm for the \texorpdfstring{$\DDT$}{DDT} problem}
\label{sec:alg_DDT}

First, we can show that an approach based on Dijkstra's algorithm
with time-dependent edge transit times (similar to~\cite{DBLP:journals/jcss/CarvalhoEP21}) yields an optimal algorithm for the variant of $\DDT_\N$ where the solution can use an arbitrary number of copies of each agent (each copy of agent $a\in A$ has the same initial location~$p_a$). 


We start by introducing some notation used in the algorithm. 
For any edge $\{u,v\}\in E_a$ for some $a\in A$, we denote by $\eT_a(v, u\prec v)$ 
the earliest time for the package to arrive at $v$ if the package is carried over the edge $\{u, v\}$  by a copy of agent~$a$. 
In addition, we use $\eT(v, u\prec v)$ to denote the earliest time for the package to arrive at node $v$ if the package is carried over the edge $\{u,v\}$  by some agent, i.e., $\eT(v, u\prec v)=\min \{\eT_a(v, u\prec v)\mid \{u,v\} \in E_a, a\in A\}$.  
For every $v\in V$, we use $\eT(v)$ to denote the earliest arrival time for the package at node $v$, i.e., $\eT(v)=\min \{\eT(v, u\prec v)\mid \{u,v\} \in E\}$. Note that the package is initially at location $s$ at time $0$, i.e., $\eT(s)=0$.    
Given a node~$v$, we denote by $\mathcal{S}(v)$ the schedule for carrying the package from $s$ to $v$, i.e., $\mathcal{S}(v)=\{(a_1, s, u_1),  (a_2, u_1, u_2), \ldots, (a_h, u_{h-1}, v)\}$. 

\begin{algorithm}[!tbh] 
\KwData{Graph $G=(V,E,\ell)$; package source node $s$ and target node $y$; agent $a$ with velocity $v_a$ and initial location $p_a$ for $a\in A$}
\KwResult{earliest arrival time for package at target location $y$, i.e., $\eT(y)$}
\Begin{compute $\dist_a(p_a,v)$ for $a\in A$ and $v\in V_a$

  
  $\eT(s)\gets 0$
  
  $\eT(v)\gets \infty$ for all $v\in V \setminus \{s\}$
  
  
  $L(v)\gets \emptyset$ for all $v\in V$ \tcp{agent bringing the package to $v$} 
  
  $\proc(v)\gets 0$ for all $v\in V$ \tcp{all nodes $v$ are unprocessed}
  
  $\prev(v)\gets \emptyset$ for all $v\in V$ \tcp{previous node on optimal package path to $v$}
  
  $Q\gets \{s\}$ \tcp{priority queue of pending nodes}
  
  \While{$Q\neq  \emptyset $}{
  
        $u\gets \arg \min \{\eT(v)\mid v\in Q\}$  \tcp{node with minimum arrival time in $Q$}
        
        $Q\gets Q \setminus \{u\}$ and $\proc(u)\gets 1$
        
        \If{$u=y$}{
              break}
              
        $t\gets \eT(u)$  \tcp{arrival time when package reaches $u$}
        
          \For{neighbors $v$ of $u$ with $\proc(v)=0$ and $A(u,v)\neq \emptyset$}{
           $\{\eT(v, u\prec v), L(v, u\prec v)\}\gets$ \textsc{NeiDelivery}$(u,v,t)$
             
             \If{$\eT(v, u\prec v)<\eT(v)$}{
                  $\eT(v)\gets \eT(v, u\prec v)$
                  
                  $L(v)\gets L(v, u\prec v)$ 
                  
                  $\prev(v)\gets \{u\}$ 
                  
                  \If{$v\notin Q$}{
                  $Q\gets Q \cup  \{v\}$ with earliest arrival time $\eT(v)$}
                }}

    }
    \Return $\eT(y)$
  } 
 \caption{Algorithm for $\DDT$}
 	\label{alg:minT}
\end{algorithm}

Our algorithm adapts the approach of a time-dependent Dijkstra's algorithm~\cite{DBLP:conf/mfcs/Bartschi0M18,DBLP:journals/jcss/CarvalhoEP21}.
Algorithm~\ref{alg:minT} shows the pseudo-code.
For each node $v\in V$, we maintain a value $\eT(v)$ that represents the current upper bound on the earliest time when the package can reach $v$ (line 3--4). Initially, we set the earliest arrival time for $s$ to $0$ and for all other nodes to $\infty$.
We maintain a priority queue $Q$ of nodes $v$ with finite $\eT(v)$ value that have not
been processed yet (line 8).
In each iteration of the while-loop, we process a node $u$ in $Q$ having the earliest arrival time (line 10), where $u=s$ in the first iteration. For each unprocessed neighbor node $v$ of~$u$, we calculate the earliest arrival time $\eT(v,u\prec v)$ at node $v$ if the package is carried over the edge between $u$ and $v$ by some agent
(and we store the identity of that agent in $L(v,u\prec v)$),
by calling the subroutine \textsc{NeiDelivery}$(u,v,t)$ (Algorithm~\ref{alg:FNDforN}). If this
earliest arrival time is smaller than $\eT(v)$, then $\eT(v)$ is updated and $v$ is inserted into $Q$ (if it is not yet in $Q$) or its priority reduced to the new value of $\eT(v)$.
The algorithm terminates and returns the value $\eT(y)$ when the node being processed is $y$ (line 12). 
The schedule $\mathcal{S}=\mathcal{S}(y)$ can be constructed in a backward manner because we store for each node $v$ the involved agent $L(v)$ and its predecessor node $\prev(v)$ (line 20 and 21). 



To compute $\eT(v, u\prec v)$ in line 17 of Algorithm~\ref{alg:minT}, we call \textsc{NeiDelivery}$(u,v,t)$ (Algorithm~\ref{alg:FNDforN}): That subroutine calculates for each agent $a\in A$
with $\{u,v\}\in E_a$, i.e., for all $a\in A(u,v)$, the time when the package reaches $v$ if that agent picks it up at $u$
and carries it over $\{u,v\}$ to~$v$. The earliest arrival time at $v$ via the edge $\{u,v\}$
is returned as $\eT(v, u\prec v)$, and the agent $a^*$ achieving that arrival time is returned
as $L(v,u\prec v)$.


\begin{algorithm}[tb] 
\KwData{Edge $\{u,v\}$, arrival time for node $u$ $\eT(u)=t$, agents $A(u,v)$} 
\KwResult{$\eT(v, u\prec v)$}

\For{$a\in A(u,v)$}
 {  $\eT_a(v, u\prec v)=\max\{t, \frac{\dist_a(p_a, u)}{v_a}\} +\frac{\ell{(u,v)}}{v_a}$
   
    }
  
 $a^*=\arg \min\{\eT_a(v, u\prec v) \mid a\in A(u,v)\}$;

$\eT(v, u\prec v) \gets \eT_{a^*}(v, u\prec v)$

  $L(v, u\prec v)\gets a^*$ 

  \Return $\eT(v, u\prec v)$ and $L(v, u\prec v)$ 
 \caption{Algorithm \textsc{NeiDelivery}$(u,v,t)$ for $\DDT_\N$}
 	\label{alg:FNDforN}
\end{algorithm}

\begin{lemma}
\label{lem:minT-NS_statement}
The following statements hold for the schedule found by Algorithm~\ref{alg:minT}.
\begin{enumerate}[label=(\roman*)] 
\item  It may happen that an agent picks up and drops off the package more than once. Each time an agent $a$ carries the package over a path of consecutive edges, a copy of the agent starts at $p_a$, travels to the start node $u$ of the path, picks up the package at time $\max\{\eT(u),\frac{\dist_a(p_a,u)}{v_a}\}$,
and then carries the package over the edges of the path.
 
\item The package is carried to each node $v\in V$ at most once, and thus the schedule carries the package over at most $|V|-1$ edges. 

\end{enumerate} 
\end{lemma}

\begin{proof}
For statement (i), note that \textsc{NeiDelivery}$(u,v,t)$ (Algorithm~\ref{alg:FNDforN}) does not take into account whether an agent in $A(u, v)$ has already been involved in the trips that bring the package to $u$ at time $\eT(u)=t$ or not. Thus, it may select an agent again, and so agents may be involved in the schedule more than once, which means that a fresh copy of the agent is used each time the agent is involved. If an agent $a$ is used on two consecutive edges, then it can continue to carry the package over the second edge after it arrives at the endpoint of the first edge, and we do not need to use a fresh copy of the agent for the second edge.

According to lines 10--11 of Algorithm~\ref{alg:minT}, the node $u$ being processed is removed
from $Q$ and never enters $Q$ again (because $\proc(u)$ is set to 1 in line 11, and only nodes
with $\proc(u)=0$ can be entered into $Q$ due to the condition in line~16).
The node $u$ can become a predecessor (stored in $\prev(v)$) of another node $v$ only
at the time when $u$ is being processed and the node $v$ has not yet been processed.
Therefore, the resulting schedule for the delivery of the package visits each node at most once.
This shows statement (ii).
\end{proof}

\begin{restatable}{lemma}{lemminTNSopt}
\label{lem:minT-NS_opt}
There is an algorithm that computes an optimal schedule in time
$O(k(m+n\log n))$
for $\DDT_\N$ under the assumption that an arbitrary number of copies of each agent can be used. The package gets delivered from $s$ to $y$ along
a simple path with at most $|V|-1$ edges.
\end{restatable}

\begin{proof}
We claim that the arrival time $\eT(u)$ for each node $u$ is minimum by the time $u$ gets processed. Obviously, the first processed node $s$ has arrival time $\eT(s)=0$. Whenever a node $u$ is processed,
the earliest arrival time for each unprocessed neighbor is updated (line 19 of Algorithm~\ref{alg:minT}) if the package can
reach that neighbor earlier via node~$u$ (line 4 in Algorithm~\ref{alg:FNDforN} identifies
the agent $a^*$ that can bring the package from $u$ to the neighbor the fastest).
At the time when a node $u$ is removed from
the priority queue $Q$ and starts to be processed, all nodes $v$ with
$\eT(v)<\eT(u)$ have already been processed, and so its value $\eT(u)$ must
be equal to the earliest time when the package can reach~$u$.
The algorithm terminates when $y$ is removed from~$Q$.
 
In lines 19--21 of Algorithm~\ref{alg:minT}, we update the arrival time $\eT(v) $ and the agent $L(v)$ as well as the predecessor node $\prev(v)$ without explicitly maintaining the schedules $\mathcal{S}(v)$. The schedule $\mathcal{S}(v)$ can be retraced from $L(\cdot)$ and $\prev(\cdot)$ since the schedule found for $\eT(v)$ visits each node in $V$ at most once, cf.\ Lemma~\ref{lem:minT-NS_statement}.  
This also shows that the package gets
delivered from $s$ to $y$ along a simple path (with at most
$|V|-1$ edges) in $G$.

We can analyze the running-time of Algorithm~\ref{alg:minT} as follows. The distances $\dist_a(p_a,v)$ can be pre-computed
by running Dijkstra's algorithm with Fibonacci heaps
in time $O(m+n\log n)$~\cite{CLRS/09} for each agent $a$ with source
node $p_a$ in the graph $G_a$, taking total time $O(k(m+n\log n))$ (line 2).
Selecting and removing a minimum element from the priority queue $Q$ in lines 10--11 takes time $O(\log n)$. At most $n$ nodes will be added into $Q$ (and later removed from $Q$), so the running time for inserting and removing elements from $Q$ is $O(n \log n)$.
 For each processed node $u$, we compute the value $\eT(v, u\prec v)$ for each unprocessed adjacent node $v$ with $A(u, v)\neq \emptyset$ in time $O(|A(u,v)|)$. Overall we get a running time of $O(k(m+n\log n)+n\log n+km)=O(k(m+n\log n))$.
 \end{proof}

If the algorithm uses an agent $a$ on consecutive edges of the
package delivery path, this can be seen as a single trip by
agent~$a$. If, however, the algorithm uses the same agent $a$
on edges that are not consecutive, then this corresponds to
using different copies of agent~$a$. This is because the algorithm
assumes that $a$ can always travel from its initial location to
the node where it picks up the package, which may be a shorter trip
than traveling from the agent's previous drop-off location to
the pick-up node. Hence, the delivery schedule produced by the
algorithm is guaranteed to be feasible only if we have arbitrarily
many copies of each agent.

Next, we can show how to convert the delivery schedule $\mathcal{S}$ with
delivery time $\mathcal{T}$ produced by
the algorithm of Lemma~\ref{lem:minT-NS_opt} into a schedule that
is feasible with a single copy per agent, and we bound the resulting increase in the
delivery time $\mathcal{T}$ to obtain an approximation algorithm for $\DDT_\N$.
The conversion consists of the repeated application of modification
steps. Each step considers the first agent $a$ that is used at least twice.
Let $a$ pick up the package at $u_{i-1}$ in its first trip and carry it from $u_{j-1}$ to $u_j$ in its last trip. We then modify the schedule so that $a$ picks up the package at $u_{i-1}$ and carries it all the way to $u_j$ along a shortest
path in $G_a$. We have $\frac{1}{v_a}\dist_a(p_a,u_{i-1})\le \mathcal{T}$ and
$\frac{1}{v_a}(\dist_a(p_a,u_{j-1})+\dist_a(u_{j-1},u_j))\le \mathcal{T}$.
By the triangle inequality, $\dist_a(u_{i-1},u_j)\le \dist_a(u_{i-1},p_a)
+\dist_a(p_a,u_{j-1})+\dist_a(u_{j-1},u_j)$. Hence, agent $a$ needs time
at most $2\mathcal{T}$ to carry the package from $u_{i-1}$ to $u_j$,
and so the delivery time increases by at most $2\mathcal{T}$. Furthermore,
we can bound the number of modification steps by $\min\{\frac{n-1}{3},k-1\}$.

\begin{restatable}{theorem}{minTNSapprox}
\label{the:minT-NS_approx}
There is a $\min\{2n/3+1/3, 2k-1\}$-approximation algorithm for $\DDT_\N$.
\end{restatable}

\begin{proof} 
Based on the solution $\mathcal{S}(y)$ obtained by the algorithm of Lemma~\ref{lem:minT-NS_opt}, we create a feasible solution $\mathcal{S}'(y)$ for $\DDT_\N$
by modifying the schedule until each agent is involved in the delivery at most once.
Initially, let $\mathcal{S}'(y)=\mathcal{S}(y)$.
Consider the first agent $a$ in $\mathcal{S}'(y)$ that is involved in the delivery at least twice.
Let $a$ be involved for the first time as the $i$-th agent and for the last time as the $j$-th agent.
Then $a$ carries the package from node $u_{i-1}$ to node $u_i$ during
its first involvement and from node $u_{j-1}$ to node $u_j$ during its last involvement in $\mathcal{S}'(y)$, and at least one other agent is involved in carrying the package from $u_i$ to $u_{j-1}$.
We modify the schedule and use agent $a$ to carry
the package from $u_{i-1}$ to $u_{j}$ along a shortest path in $G_a$ and
remove the trips made by the $(i+1)$-th agent to the $(j-1)$-th agent.
If the package arrives at $u_{j}$ later in the resulting schedule,
the times of all subsequent trips are delayed accordingly. Replace
$\mathcal{S}'(y)$ by the new schedule, and repeat this operation as long
as there exists an agent that is involved in the delivery at least twice.

Note that the schedule $\mathcal{S}(y)$ transports the package over a path $P$ consisting of at most $n-1$ edges. 
A modification step that lets $a$ carry the package from $u_{i-1}$ to $u_j$ directly
can be charged to three edges of $P$ that are not charged by any other modification
step: An edge on $P$ between $u_{i-1}$ and $u_i$, one between $u_i$ and $u_{j-1}$, and
one between $u_{j-1}$ and $u_j$. Therefore, there are at most $(n-1)/3$ modification steps.
Similarly, the agent $a$ for which a modification step is carried out will not be
involved in another modification step, and modification steps can only be necessary
if there are still at least two agents for which no modification step has been
carried out so far. Therefore, the number of modification steps can also be
bounded by $k-1$.
Thus, to obtain the final $S'(y)$, we need to modify the schedule at most $\min\{(n-1)/3, k-1\}$ times.

Next, we bound the delivery time of the constructed schedule $\mathcal{S}'(y)$.
We claim that each modification step increases the delivery time of the package
by at most $2\eT(y)$.
Consider a modification step that lets agent $a$ carry the package from $u_{i-1}$
to $u_{j}$. Considering the original schedule $\mathcal{S}(y)$, we observe:
$$
\eT(u_i)\ge \frac{\dist_a(p_a,u_{i-1})}{v_a}+\frac{\dist_a(u_{i-1},u_i)}{v_a}
$$
$$
\eT(u_j)\ge \frac{\dist_a(p_a,u_{j-1})}{v_a}+\frac{\dist_a(u_{j-1},u_j)}{v_a}
$$
By the triangle inequality, we have
$\dist_a(u_{i-1},u_j) \le
\dist_a(u_{i-1},p_a)+\dist_a(p_a,u_{j-1})+\dist_a(u_{j-1},u_j)$
and hence
\begin{align}
\frac{\dist_a(u_{i-1},u_j)}{v_a} & \le
\frac{\dist_a(u_{i-1},p_a)}{v_a}+\frac{\dist_a(p_a,u_{j-1})}{v_a}+\frac{\dist_a(u_{j-1},u_j)}{v_a}\\
& \le \eT(u_i)+\eT(u_j) \le 2\cdot \eT(y)
\end{align}
The time when $a$ picks up the package at node $u_{i-1}$ is unchanged
by the modification, and the time for carrying the package from $u_{i-1}$
to $u_j$ is at most $2\cdot\eT(y)$ after the modification. Therefore,
the time when the package reaches $u_j$, and hence the delivery time
of the package at $y$, increases by at most $2\cdot\eT(y)$ in the modification, as claimed.

Since there are at most $\min\{(n-1)/3, k-1\}$ modification steps, we have  $$\eT'(y)\le \eT(y)+ \min\{(n-1)/3, k-1\} \cdot 2\cdot \eT(y)=\min\{2n/3+1/3, 2k-1\} \cdot \eT(y).$$ This completes the proof.
 \end{proof}

%

To adapt the approach from $\DDT_\N$ to $\DDT_\E$, we can
adapt the algorithm of Lemma~\ref{lem:minT-NS_opt} to
edge handovers by using as a subroutine the
\textsc{FastLineDelivery}$(u,v,t)$ method
from~\cite{DBLP:journals/jcss/CarvalhoEP21}, which calculates
in $O(k\log k)$ time an optimal delivery schedule using the
agents in $A(u,v)$ to transport the package that
arrives at $u$ at time $t$ from $u$ to $v$ over the
edge $\{u,v\}$.
When transforming the resulting package delivery schedule
into one that uses each agent at most once,
the number of modification steps can be bounded by $\min\{n-1,k-1\}$.

 For $\DDT_\E$, we may need to choose a set of agents to collaborate on carrying the package over an edge $\{u,v\}$ since the package can be handed over at any point in the interior of the edge. Therefore, in line 17 of Algorithm~\ref{alg:minT} we now call as subroutine \textsc{NeiDelivery}$(u,v,t)$ instead
of Algorithm~\ref{alg:FNDforN} a
subroutine that computes the fastest way for carrying the package from $u$ to $v$ over
the edge $\{u,v\}$ assuming that the package has arrived at node $u$ at time~$t$.
This subproblem can be solved in $O(k\log k)$ time using the algorithm
\textsc{FastLineDelivery}$(u,v,t)$ from~\cite{DBLP:journals/jcss/CarvalhoEP21}.
That algorithm can be applied in our setting with movement restrictions because
it considers a single edge $\{u,v\}$. Only the agents in $A(u,v)$ are relevant
for the solution, and all other agents can be ignored. We know for
each agent in $A(u,v)$ the earliest time when it can reach $u$ and the earliest
time when it can reach $v$, and so we can execute \textsc{FastLineDelivery}$(u,v,t)$
(together with its subroutines \textsc{PreprocessSender} and \textsc{PreprocessReceiver})
from~\cite{DBLP:journals/jcss/CarvalhoEP21} in $O(k\log k)$ time.

Similar to Lemma~\ref{lem:minT-NS_opt}, we thus get the following:

 \begin{lemma}
\label{lem:minT-ES_opt}
Algorithm~\ref{alg:minT} computes an optimal schedule for $\DDT_\E$ in time $O(k(m\log k +n\log n))$ under the assumption that an arbitrary number of copies of each agent can be used.   
 \end{lemma} 
 

We note that the schedule produced for agents collaborating to transport the
package over an edge $\{u,v\}$ by algorithm \textsc{FastLineDelivery}$(u,v,t)$ from~\cite{DBLP:journals/jcss/CarvalhoEP21} has the property that the velocities
of the agents used are strictly increasing. Therefore, each agent is used at most
once on each edge of the path along which the package is delivered from $s$ to $y$.
We can now show the following theorem:

\begin{restatable}{theorem}{minTESapprox}
\label{the:minT-ES_approx}
There is a $\min\{2n-1, 2k-1\}$-approximation algorithm for $\DDT_\E$.
\end{restatable}

\begin{proof}
Similar to the proof of Theorem~\ref{the:minT-NS_approx}, we repeatedly modify the schedule $\mathcal{S}$ obtained by Algorithm~\ref{alg:minT} (using \textsc{FastLineDelivery} from~\cite{DBLP:journals/jcss/CarvalhoEP21} as
the \textsc{NeiDelivery} subroutine) for $\DDT_\E$ until each agent is involved at most once. The difference here is that every modification can only be charged to a single edge, not three edges.
For example, if the modification operation is applied to an agent $a$ that carries the package
on parts of consecutive edges $e_1$ and $e_2$, then the modification can only be charged to $e_1$,
because another agent carrying the package on a later part of $e_2$ may be involved in
another modification. Thus, the number of modification steps can only be bounded by
$\min\{n-1, k-1\}$. By similar arguments as in the proof of Theorem~\ref{the:minT-NS_approx}, we get $\eT'(y)\le \min\{2n-1, 2k-1\} \cdot \eT(y)$. 
\end{proof}


\section{Approximation algorithm for the \texorpdfstring{$\DDC$}{DDC} problem}
\label{sec:alg_DDC}%
By Lemma~\ref{lem:eqMinC-SE-SN}, handovers on interior points of edges
are not needed for $\DDC$, so the results for $\DDC_\N$ that
we present in this section automatically apply to $\DDC_\E$ as well.
Therefore, we only consider $\DDC_\N$ in the proofs. We first give an
algorithm that solves $\DDC_\N$ optimally if there is a sufficient
number of copies of every agent.

Let an instance of $\DDC_\N$ be given by
a graph $G=(V,E,\ell)$, package start node $s$ and destination node $y$, and a set $A$ of $k$ agents where $a=(p_a, v_a, w_a, V_a, E_a)$ for $a\in A$.
We create a directed graph $G'$ in which a shortest
path from $s'$ to $y'$ corresponds to an optimal delivery schedule.
This approach is motivated by a method used by B{\"{a}}rtschi et al.~\cite{DBLP:conf/stacs/BartschiC0D0HP17}.
We construct the directed edge-weighted graph $G'=(V', E', \ell')$ as follows:

\begin{itemize}
    \item For each node $u\in V$ and each agent $a \in A$ with $u\in V_a$, create a node $u_a$ in $V'$. In addition, add a node $s'$ and a node $y'$.
    
    \item For all $a\in A$ with $s\in V_a$, add an arc $(s', s_a)$ with $\ell'{(s', s_a)}=w_a\cdot \dist_a(p_a, s)$. For all $a\in A$ with $y\in V_a$, add an arc $(y_a, y')$ with $\ell'{(y_a, y')}=0$.
    
    \item For $\{u,x\}\in E$, for each $a$ with $\{u,x\} \in E_a $, create two arcs $(u_a, x_a)$ and $(x_a,u_a)$ with $\ell'{(x_a,u_a)}=\ell'{(u_a,x_a)}=w_a\cdot \ell{(u,x)}$.
    
    \item For $u\in V$ and agents $a,\bar{a}\in A(u)$, create the following two arcs: $(u_a, u_{\bar{a}})$ with $\ell'{(u_a, u_{\bar{a}})}=w_{\bar{a}}\cdot \dist_{\bar{a}}(p_{\bar{a}}, u)$, and  $(u_{\bar{a}}, u_a)$ with $\ell'{(u_{\bar{a}}, u_a)}=w_a\cdot \dist_a(p_a, u)$.
\end{itemize}
Intuitively, a node $u_a$ in $G'$ represents the agent $a$ carrying the package
at node $u$ in $G$. An arc $(u_a,x_a)$ represent the agent $a$ carrying the
package over edge $\{u,x\}$ from $u$ to $x$. An arc $(u_a,u_{\bar{a}})$
represents a copy of agent $\bar{a}$ traveling from $p_{\bar{a}}$ to
$u$ and taking over the package from agent $a$ there.
We can show that a shortest $s'$-$y'$ path in $G'$ corresponds
to an optimal schedule with multiple copies per agent.

\begin{restatable}{lemma}{DDCmulti}
\label{lem:DDCmulti}%
An optimal schedule for $\DDC_\N$ (and $\DDC_\E$) can be computed
in time $O(nk^2+n^2k)$
under the assumption that an arbitrary number of copies of each agent
can be used.
\end{restatable}

\begin{proof}
We claim that a shortest $s'$-$y'$ path in $G'$ corresponds to an optimal
delivery schedule. First, assume that an optimal delivery schedule $\mathcal{S}$ is
$$\{(a_1, s, u_1), (a_2, u_1, u_2), \ldots, (a_{h}, u_{h-1}, y)\},$$ 
where it is possible that $a_i=a_j$ for $j>i+1$ because we allow
copies of agents to be used. Then we can construct an $s'$-$y'$ path
in $G'$ whose length equals the total energy consumption of $\mathcal{S}$
as follows: Start with the arc $(s',s_{a_1})$. Then use the arcs
$(s_{a_1},z^{(1)}_{a_1}), \ldots, (z^{(g)}_{a_1},{u_1}_{a_1})$
corresponding to the path $s,z^{(1)},\ldots,z^{(g)},u_1$ (for some $g\ge 0$) along
which agent $a_1$ carries the package from $s$ to $u_1$ in $\mathcal{S}$.
Next, use the arc $({u_1}_{a_1},{u_1}_{a_2})$ representing the handover
from $a_1$ to $a_2$ at $u_1$. Continue in this way until node
$y_{a_{h}}$ is reached, and then follow the arc from there to $y'$.
Similarly, any $s'$-$y'$ path $P$ in $G'$ can be translated into a
delivery schedule in $G$ whose total energy consumption is equal to
the length of $P$ in $G'$.

Finally, let us consider the running-time of the algorithm.
First, we compute $\dist_a(p_a,v)$ for each agent $a$ and
each node $v\in V$ by running Dijkstra's algorithm with Fibonacci heaps~\cite{CLRS/09}
once in $G_a$ with source node $p_a$ for each $a\in A$, taking
time $O(k(n\log n + m))=O(kn^2)$.
The graph $G'$ has at most $n\cdot k+2\in O(nk)$ vertices and
at most $2k+(n^2\cdot k+ n\cdot k^2)\in O(nk^2+n^2k)$ arcs.
It can be constructed in $O(nk^2+n^2k)$ time as we have
pre-computed the values $\dist_a(p_a,v)$.
We can compute the shortest $s'$-$y'$ path in $G'$ in time $O(nk^2+n^2k+nk\log(nk))= O(nk^2+n^2k)$ time using Dijkstra's algorithm.
\end{proof}

\begin{theorem}
\label{the:minC_approx}
There is a 2-approximation algorithm for $\DDC_\N$ (and $\DDC_\E$).
\end{theorem}

\begin{proof}
Let an instance of $\DDC_\N$ be given by
a graph $G=(V,E,\ell)$, package start node $s$ and destination node $y$, and a set $A$ of $k$ agents where $a=(p_a, v_a, w_a, V_a, E_a)$ for $a\in A$.
Compute an optimal delivery schedule that may use multiple copies of agents
using the algorithm of Lemma~\ref{lem:DDCmulti}.
Then we transform the schedule into one that uses each agent at most once
as follows: Let $a$ be the first agent that is used more than once in
the delivery schedule. Assume that $a$ carries the package
from node $u$ to node $v$ during its first involvement in the delivery and
from node $u'$ to node $v'$ during its last involvement in the delivery.
The energy consumed by these two copies of agent $a$ is:
$$
W_a = w_a ( \dist_a(p_a,u)+\dist_a(u,v)+\dist_a(p_a,u')+\dist_a(u',v') )
$$
We modify the schedule and let agent $a$ pick up the package at $u$
and carry it along a shortest path in $G_a$ from $u$ to $v'$. The
trips in the original schedule that bring the package from $v$
to $u'$ are removed. The new energy consumption by agent $a$ is
$w_a (\dist_a(p_a,u)+\dist_a(u,v'))$. By the triangle inequality,
$\dist_a(u,v')\le \dist_a(u,p_a)+\dist_a(p_a,u')+\dist_a(u',v')$.
Hence, the new energy consumption is bounded by $2W_a$.
As long as there is an agent that is used more than once, we
apply the same modification step to the first such agent.
The procedure terminates after at most $k-1$ modification steps.
During the execution of the procedure, the energy consumption
of every agent at most doubles. Therefore, the total energy
consumption of the resulting schedule for $\DDC_\N$ with a single
copy of each agent is at most twice the total energy consumption
of $\mathcal{S}$, which is a lower bound on the optimal energy
consumption.

The algorithm of Lemma~\ref{lem:DDCmulti} takes $O(nk^2+n^2k)$ time,
which dominates the time needed to carry out the modification steps.
Thus, the overall running time is $O(nk^2+n^2k)$.
\end{proof}

\section{Drone delivery on path and tree networks}
\label{sec:line}%

\subsection{Hardness of \texorpdfstring{$\DDT$}{DDT} on the path}
\label{sec:line-hardness}%

\begin{theorem}
$\DDT_\N$ and $\DDT_\E$ with initial positions are $\NP$-complete if the given graph is a path.
\end{theorem}

\begin{proof}
We give a reduction from the 
$\NP$-complete \textsc{Even-Odd Partition} (EOP) problem \cite{DBLP:books/fm/GareyJ79}
that is defined as follows:
Given a set of $2n$ positive integers $X=\{x_1, x_2, \ldots, x_{2n}\}$, is there a partition of $X$ into two subsets $X_1$ and $X_2$ such that $\sum_{x_i\in X_1} x_i= \sum_{x_i\in X_2} x_i$ and such that, for each $i\in 1, 2, \ldots, n$, $X_1$ (and also $X_2$) contains exactly one element of $\{x_{2i-1}, x_{2i}\}$?
Let an instance of EOP be given by $2n$ numbers $\{x_1, x_2, \ldots, x_{2n-1}, x_{2n}\}$.  
Construct a path (see Figure~\ref{fig:hardnessline}) with set of nodes $\{s, b_1, b_2, \ldots, b_{2n}, z, z', c_{2n}, c_{2n-1}, \ldots, c_{1}, y, u\}$  ordered from left to right.
The length of each edge corresponds to the Euclidean distance between
its endpoints.
We let $z'-z=b_{2i}-b_{2i-1}=1$, $c_{2i-1}-c_{2i}=1/2$, $z-b_{2n}=c_{2n}-z'=y-c_1=L$, and $b_{2i+1}-b_{2i}=c_{2i}-c_{2i+1} =L$ for all $i\in [n-1]$ where $L=3Cn/2+7/4+1$.
Furthermore, we let $b_1-s=S$, where $S>0$ is specified later.
The node $u$ is placed so that $u-z=b_1-s+Cn+0.5$, where $C$ is a constant that we can set to $C=3$.
There are $4n+3$ agents:
\begin{itemize}
     \item Two agents $h_1, h_2$: These agents have speed~$1$.  Agent $h_1$ is initially at node $s$ and can traverse the interval $[s, b_1]$.
    Agent $h_2$ is initially at node $u$ and can traverse the interval $[z, u]$. 
     \item $2n$ agents $p_1, p_2, \ldots, p_{2n}$: These agents are initially at node $z$; each agent $p_i$ has speed $v_{p_i}=\frac{1}{C+x_i/M}$ where $M=\sum_{i\in [2n]} x_i$.
     Agents $p_{2i-1}$ and $p_{2i}$ for $i\in [n]$ can traverse the interval $[b_{2i-1}, c_{2i-1}]$.   
    \item $2n+1$ agents $f_1, f_2, \ldots, f_{2n}, f_{2n+1}$: These agents have infinite speed. Each agent $f_{i}$ for $i<n$ is initially at node $b_{2i}$ and can traverse the interval $[b_{2i}, b_{2i+1}]$; agent $f_n$ is initially at node $b_{2n}$ and can traverse the interval $[b_{2n}, z]$;   agent $f_{n+1}$ is initially at node $z'$ and can traverse the interval $[z', c_{2n}]$;  each agent $f_{i}$ for $n+1<i\le 2n$ is initially at node $c_{2(2n+1-i)+1}$ and can traverse the interval $[c_{2(2n+1-i)+1}, c_{2(2n+1-i)}]$; agent $f_{2n+1}$ is initially at node $c_1$ and can traverse the interval $[c_1, y]$.
    \end{itemize}

\begin{figure} 
\centering
\includegraphics[scale=0.75]{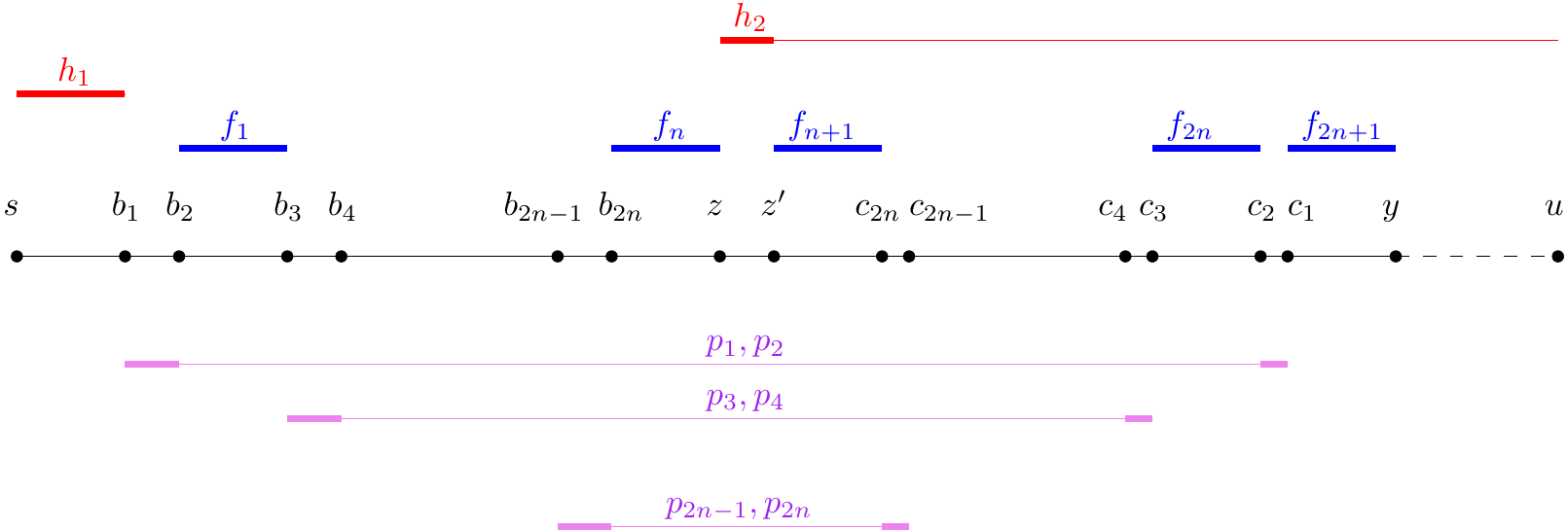}
\caption{Instance of $\DDT_\N$ on a path}
\label{fig:hardnessline}
\end{figure}

The goal is to deliver a package from $s$ to $y$ as quickly as possible. 
We claim that there is a schedule with delivery time at most $T =S+3Cn/2+7/4$
if and only if the original instance of EOP is a yes-instance.

First, assume that the instance of EOP is a yes-instance with partition $(X_1,X_2)$.
Let $(P_1,P_2)$ be the corresponding partition of the set containing the $2n$ agents $p_i$.
We construct a delivery schedule for the package as follows (the colors we mention refer to
those shown in Figure~\ref{fig:hardnessline}).
Until time $S=b_1-s$, the agent $h_1$ carries the package to $b_1$, and all agents $p_i$ pick a bold purple interval and arrive at its left endpoint. This is done in such a way that agent $p_i$ picks
its left bold purple interval (the interval of length 1 at the left end of its range)
if $p_i\in P_1$, and its right bold purple interval (the interval of length $\frac12$ at the
right end of its range) otherwise.
To guarantee that any agent $p_i$ arrives at the left endpoint of its interval by time $S$, we set $S= \max_{i\in [n]} \frac{(n+1-i)(L+1)}{\min\{p_{2i-1}, p_{2i}\}}$.  
Next, the package travels from $b_1$ to $z$ by always being alternatingly carried by an agent $p_i$ over a bold purple interval of length 1 and an infinite speed agent $f$ over a blue interval of length $L$.
The package thus reaches $z$ at time $S + Cn + 1/2$ (because $\sum_{x_i \in X_1} x_i/M=1/2$).
As $u-z = S+Cn + 1/2$, the agent $h_2$ arrives at $z$ at exactly the same time.
Then, the agent $h_2$ carries the package from $z$ to $z'$ in time $1$.
After that, the package is alternatingly carried by an infinite speed agent and an agent $p_i$ until it reaches $y$, taking time $Cn/2 + 1/4$.
The total delivery time is $S+Cn+1/2+1+Cn/2+1/4=S+3Cn/2+7/4=T$ as required.
 
Now consider the case that the original EOP instance is a no-instance. Assume that
there exists a delivery schedule with delivery time $T=S+3Cn/2+7/4$.
First, we claim that neither an agent $p_i$ nor the agent $h_2$ can carry the package over an interval of length~$L$ (these are the intervals $[b_{2n}, z]$, $[z', c_{2n}]$, $[c_1, y]$, and the intervals $[b_{2i}, b_{2i+1}]$, $[c_{2i+1}, c_{2i}]$ for any $i\in [n-1]$). Otherwise, the delivery time is larger than $S+L> S+3Cn/2+7/4$ because these agents have speed at most~$1$ and these blue intervals have length $L=3Cn/2+7/4+1$. Therefore, it
is clear that each of the $2n$ agents $p_i$ and the agent $h_2$ will carry the package over exactly one of the following $2n+1$ intervals: the interval $[b_1,b_2]$, and the $2n$ intervals that lie between two consecutive blue intervals. Next, we claim that $h_2$ must carry the package over
the interval $[z,z']$ with length~$1$. Otherwise, $h_2$ would have to carry the package
over an interval $[c_{2i},c_{2i-1}]$ with length $1/2$, and an agent $p_i$ with speed
less than $1/C$ would have to carry the package through $[z,z']$ with length $1$ instead. Even if we assume that all agents $p_i$ have
the faster speed $1/C$, the resulting schedule would have delivery time
at least $S+nC+C+(n-1)C/2+1/2= S+3Cn/2+C/2+1/2$, which is larger than
$T=S+3Cn/2+7/4$ as $C=3>5/2$. 

This implies that agent $h_2$ carries the package over $[z,z']$ and,
for each $i\in [n]$, one agent among $\{p_{2i-1}, p_{2i}\}$ carries the package on a bold purple interval to the left of $z$ and the other agent carries the package on a bold purple interval to the right of $z'$.  
Let the resulting partition of the agents $p_i$ be $(P_1,P_2)$, and
let the corresponding partition of the original EOP instance be $(X_1,X_2)$.
Suppose the package arrives at node $z$ at time $S+W$.  Since the instance of EOP is a no-instance,
either $W>Cn+1/2$ or $W<Cn+1/2$ holds. As the agents $p_i$ that carry the package to the left of $z$ take time $W$ in total, the agents $p_i$ that carry the package over intervals to the right of $z'$
take time $\frac{2Cn+1-W}{2}$ in total. 
Consider the two cases:
\begin{itemize}
    \item $W> Cn+1/2$. As agent $h_2$ carries the package
    from $z$ to $z'$, the delivery time is
    $S+W+1+\frac{2Cn+1-W}{2}=S+Cn+W/2+3/2> S+3Cn/2+7/4=T$, a contradiction.
    \item $W< Cn+1/2$. As agent $h_2$ carries the package
    from $z$ to $z'$, the package must wait at $z$ until time
    $S+Cn+1/2$ when $h_2$ reaches $z$.
    The delivery time is $S+Cn+1/2+1+\frac{2Cn+1-W}{2}
    =S+2Cn+2-W/2 > S+2Cn+2-(Cn+1/2)/2
    =S+3Cn/2+7/4=T$, a contradiction.
    \end{itemize}

Finally, we observe that in the constructed instance of $\DDT$, the availability of
edge handovers has no impact on the existence of a schedule with delivery time $S+3Cn/2+7/4$.
Therefore, the reduction establishes $\NP$-hardness of both $\DDT_\N$ and $\DDT_\E$.
\end{proof}

\subsection{Algorithms for drone delivery on a tree}
\label{sec:line-alg}%
 In this section, we show that all variants of the drone delivery
 problem can be solved optimally in polynomial time if the graph
 is a tree and all agents have the same speed (for minimizing the
 delivery time) or the same energy consumption rate (for minimizing
 the total energy consumption). In fact, the algorithms extend
 to the case of general graphs if the subgraph $G_a=(V_a,E_a)$
 of each agent is \emph{isometric}, i.e., it
 satisfies the following condition: For any two
 nodes $u,v\in V_a$, the length of the shortest $u$-$v$ path in $G_a$
 is equal to the length of the shortest $u$-$v$ path in~$G$.
 If the given graph is a tree, the subgraph $G_a$ of every agent $a$
 is necessarily isometric because we assume that $G_a$ is connected
 and there is a unique path between any two nodes in a tree.
 
 If all agents have the same speed (for $\DDT$) or the same
 consumption rate (for $\DDC$), handovers at internal points
 of edges can never improve the objective value. Therefore,
 we only need to consider $\DDT_\N$ and $\DDC_\N$ in the following.
 
 The crucial ingredient of our algorithms for the case of
 isometric subgraphs $G_a$ is:
 
 \begin{restatable}{lemma}{lemisometric}
 \label{lem:isometric}%
 Consider a delivery schedule $\mathcal{S}$ that may use an arbitrary number
 of copies of each agent. If the subgraph of each agent is
 isometric and all agents have the same speed (or the same
 energy consumption rate), then $\mathcal{S}$ can be transformed
 in time $O(k(m+n\log n))$ into a schedule in which each agent is used at most once
 without increasing the delivery time (or the energy consumption).
 \end{restatable}

\begin{proof}
Consider the first agent $a$ that is used at least twice in $\mathcal{S}$.
Assume that it carries the package from $u_{i-1}$ to $u_i$ in its first
trip and from $u_{j-1}$ to $u_j$ in its last trip. Change the schedule
so that $a$ carries the package from $u_{i-1}$ to $u_j$, and discard
the trips by agents $i+1,\ldots,j-1$. As $G_a$ is isometric, the
agent~$a$ carries the package from $u_{i-1}$ to $u_j$ along a shortest
path in $G$, and hence neither the time (in case of equal speed)
nor the energy consumption (in case of equal energy consumption rate)
increase by this modification. Repeat the modification step until every
agent is used at most once.

There are at most $k-1$ modification steps, and each of them can
be implemented in $O(m+n\log n)$ time using Dijkstra's algorithm
with Fibonacci heaps~\cite{CLRS/09}.
\end{proof}

 \begin{theorem}
 $\DDT_\N$ (and $\DDT_\E$) can be solved optimally in time
 $O(k(n\log n+m))$ if all agents have the same speed and the
 subgraph of every agent is isometric.
 \end{theorem}
 
 \begin{proof}
 Compute an optimal delivery schedule that may use multiple copies
 of each agent using Lemma~\ref{lem:minT-NS_opt}
 and then apply  Lemma~\ref{lem:isometric}.
 \end{proof}
 
 \begin{theorem}
 $\DDC_\N$ (and $\DDC_\E$) can be solved optimally in time
 $O(nk^2+n^2k)$ if all agents have the same speed and the
 subgraph of every agent is isometric.
 \end{theorem}

\begin{proof}
Compute an optimal delivery schedule that may use multiple
copies of each agent using Lemma~\ref{lem:DDCmulti} and then
apply Lemma~\ref{lem:isometric}.
\end{proof}

The problem variants without initial positions can
also be solved optimally in polynomial time:
We simply compute a shortest $s$-$y$-path $P$ and
place on each edge $e$ of $P$ a copy of an arbitrary agent
in $A(e)$ and then apply Lemma~\ref{lem:isometric}.

For the special case where $G$ is a tree, the running-time
for $\DDT_\N$ can be improved to $O(kn)$ by using a simple
algorithm that can even be implemented in a distributed way:
The package acts as a magnet, and each agent moves towards
the package until it meets the package and then follows it
(or carries it) towards $y$ as long as its range allows.
When several agents are at the same location as the package,
the one whose range extends furthest towards $y$ carries
the package. 

For the case where the given graph $G$ is a tree, we first have the
following observation:

\begin{observation}
\label{obs:tree_unit_path}
For any instance of the $\DDT$ or $\DDC$ problem on a tree, there is an optimal solution in which the package is carried along the unique simple path $\Path(s,y)$ between start node $s$ and destination node $y$.
\end{observation} 

\begin{proof}
If the package was ever transported away from $\Path(s,y)$, it would
have to be returned to $\Path(s,y)$ before moving closer to $y$ on $\Path(s,y)$,
and this detour could be removed from the solution without increasing
the delivery time or energy consumption.
\end{proof}


\begin{lemma}
For any instance of $\DDT$ on a tree, if all agents have the same speed, we can find an optimal solution for both $\DDT_\N$ and $\DDT_\E$ with initial positions in time $O(kn)$. Furthermore, we can find an optimal solution for both $\DDT_\N$ and $\DDT_\E$ without initial positions in time $O(kn)$.  
\end{lemma}

\begin{proof}  
Consider this algorithm for $\DDT_\N$ with initial positions:
\begin{itemize}
    \item At any moment in time, starting at $t=0$, each agent $a$ moves from its initial location $p_a$ towards the current location of the package; it only stops moving when it hits an endpoint of its node range. 
    
    \item At each moment in time $t$ when there is a set of agents that is at the same location as the package and at least one of them can carry
    the package in the direction of $y$, let $D(t)$ be the set
    of agents that can carry the package further in the direction of $y$. Let the agent $a$ in $D(t)$ that contains the node $b(a)$ closest
    to $y$ in $V_a \cap V(\Path(s, y))$ take the package and carry it
    towards $b(a)$ until it either reaches $b(a)$ or the package can be transferred at a node to another agent $a'$ whose node $b(a')$ is closer to $y$ than $b(a)$.
\end{itemize}
In this algorithm, the package acts like a magnet and all agents act like pieces of metal that are attracted by the magnet. When the package is ``touched'' by an agent, it moves with unit speed in this agent's edge-range in the package path until it is handed over to another agent, which happens when it is not possible for the agent to carry the package further or if there is another agent at the location of the package whose edge range is larger. We remark that this
algorithm could be implemented in a distributed manner as each agent
only needs to know the current location of the package to determine its movement.

We denote the nodes where the package is picked up by the first agent or handed over from one agent
to another agent or delivered to $y$ by $H=\{u_1, u_2, ... , u_h\}$ with $|H|=h$ in order from $u_1=s$ to $u_h=y$. Obviously, the package is handed to a different agent at each handover node.
Let agent $i$ be the agent that carries the package from $u_{i-1}$ to $u_i$,
for $2\le i\le h$.
Next, we show that the package is delivered to $u_{h}=y$ with minimum delivery time. 
 
We can prove that the package reaches each node of $H$ as early as possible by induction. 
Let $P(u_i)$ be the statement that the package reaches $u_i$ as early as possible.   
The base case $P(u_1)$ is obvious, as the package is at $u_1=s$ at time $0$. Next, we show that for every node $u_i \in H$ with $i<h$, if $P(u_i)$ holds, then $P(u_{i+1})$ also holds. Note that the package reaches $u_{i+1}$ at the time that is equal to the  maximum of the package's arrival time at node $u_i$ and agent $i+1$'s arrival time at $u_i$, plus the travel time from $u_i$ to $u_{i+1}$. By the inductive hypothesis, the package is carried to $u_i$ with minimum arrival time, and the travel time  from $u_i$ to $u_{i+1}$ is fixed because the speed of all agents is the same. So we only need to show that the package is handed at $u_i$ to an agent that reaches $u_i$ as early as possible among all agents that
can carry the package from $u_i$ closer towards $y$.

Assume that $a^*$ is the agent that can reach $u_i$ as early
as possible and that can carry the package closer towards $y$.
In the algorithm, agent $a^*$ first moves from its initial location
to the closest node $v$ on $\Path(s,y)$. At the time when
$a^*$ reaches $v$, there are two possibilities: If the package
is already closer to $y$ than $v$, then $a^*$ will chase the
package and reach node $u_{i}$ as early as possible.
If the package is further from $y$ than $v$, then $a^*$ will
first move towards the package until it is in the same location
and then travel along with the package at the same speed until
the package reaches $u_{i}$. Thus, $a^*$ is guaranteed to
reach $u_i$ by the time when the package reaches $u_i$ or
as early as possible afterwards. Thus, $P(u_{i+1})$ holds
as well.


Note that the proof for $\DDT_\E$ is the same, because handovers in interior points of edges cannot reduce the delivery time if all agents have
the same speed.

For the setting without initial locations, we choose for
each agent $a\in A$ with $V(\Path(s, y)) \cap V_a \neq \emptyset$ as
initial location $p_a$ the node on $V(\Path(s, y)) \cap V_a$ that is closest to~$s$. The remaining agents (if any) do not participate in the delivery of
the package, and their initial locations can be chosen arbitrarily.
It is easy to see that the delivery time is then equal to
the length of $\Path(s,y)$ divided by the speed of the agents, which is clearly optimal.

The algorithms for all problem variants can be implemented in $O(kn)$
time using standard techniques. In particular, for each agent
$a\in A$, we can in $O(n)$ time determine for each node on
$\Path(s,y)$ the earliest time when $a$ can reach that node.
Whenever the package reaches a node $v$ on $\Path(s,y)$,
we can then in $O(k)$ time identify the earliest agent that
can carry the package from $v$ to the next node in the path.
\end{proof}

Furthermore, we have: 
\begin{corollary}
We can determine an optimal solution for both $\DDC_\N$ and $\DDC_\E$ without initial positions in time $O(kn)$ on a tree if all agents have the same consumption rate. 
\end{corollary}   

\begin{proof}
Since the package is delivered following a fixed path, the energy consumption of the optimal solution for both $\DDC_\N$ and $\DDC_\E$ without initial locations is simply equal to the total energy consumption of carrying the package over the fixed path, i.e., $\frac{\sum_{e\in E(\Path(s, y))} \ell(e)}{w}$, where $w$ is the consumption rate of the agents. The schedule is the same as an optimal solution for $\DDT_\N$ or $\DDT_\E$ without initial locations.
\end{proof}

\section{Conclusions}
\label{sec:conclusion}%
In this paper we have studied drone delivery problems
in a setting where the movement area of each drone is
restricted to a subgraph of the whole graph. For $\DDT$,
we have presented a strong inapproximability result
and given a matching approximation algorithm.
For $\DDC$, we have shown $\NP$-hardness and presented
a $2$-approximation algorithm. For the interesting special case
of a path (line), we have shown that $\DDT$ is $\NP$-hard if the agents
can have different speeds.
For trees (or, more generally, the case where the subgraph
of each agent is isometric), we have shown that all problem variants can
be solved optimally in polynomial time if the agents have
the same speed or the same energy consumption.

We leave open the complexity of $\DDC$ on a path.
For the case without initial positions, the complexity
of both $\DDC$ and $\DDT$ on a path remains open.
For $\DDT$ with initial positions on the path,
a very interesting question is how well the problem
can be approximated.
    


\bibliography{main}

    \bibliographystyle{alpha}
    





\end{document}